\documentclass{article}

\usepackage[margin=1in]{geometry}
\usepackage{mycommon}
\DeclareMathOperator{\disc}{disc}
\DeclareMathOperator{\herdisc}{herdisc}
\DeclareMathOperator{\col}{col}
\newcommand{\uni}{\mathcal{U}}
\newcommand{\range}{\mathcal{R}}

\begin{document}

\title{Optimal Private Halfspace Counting via Discrepancy}

\author{S. Muthukrishnan\thanks{Rutgers University,
    muthu@cs.rutgers.edu} \and Aleksandar Nikolov\thanks{Rutgers
    University anikolov@cs.rutgers.edu}}

\maketitle
\begin{abstract}
A \emph{range counting} problem is specified by a set $P$ of size
$|P| = n$ of points in $\mathbb{R}^d$,  an integer {\em weight} $x_p$
associated to each point $p \in P$,
and a \emph{range space} $\range \subseteq 2^{P}$. Given 
a query range $R \in \range$, the output is  $R(\vec{x}) = \sum_{p \in R}{x_p}$. 
%For  range counting problems,  
The \emph{average squared error} 
of an algorithm $\alg$ is
$  \frac{1}{|\range|}\sum_{R \in \range}{\left(\alg(R, \vec{x}) - R(\vec{x})\right)^2}$. 
Range counting for different range spaces is a central problem in Computational Geometry. 

We study $(\eps, \delta)$-differentially private algorithms for range
counting.  Our main results are for the range space given by
hyperplanes, that is, the halfspace counting problem. We present an
$(\eps, \delta)$-differentially private algorithm for halfspace
counting in $d$ dimensions which is $O(n^{1-1/d})$ approximate for
average squared error. This contrasts with the $\Omega(n)$ lower bound
established by the classical result of Dinur and
Nissim~\cite{Dinur2003} on approximation
% of any $(\eps, \delta)$-differentially   private algorithm 
for arbitrary subset counting queries.  We also show a matching lower bound of $\Omega(n^{1-1/d})$ approximation for any $(\eps, \delta)$-differentially  private algorithm for halfspace counting. 

Both bounds are obtained using discrepancy theory. For the lower
bound, we use a modified discrepancy measure and bound approximation
of $(\eps, \delta)$-differentially private algorithms for range
counting queries in terms of this discrepancy. We also relate the
modified discrepancy measure to classical combinatorial discrepancy,
which allows us to exploit known discrepancy lower bounds. This
approach also yields a lower bound of $\Omega((\log n)^{d-1})$ for
$(\eps, \delta)$-differentially private {\em orthogonal} range
counting in $d$ dimensions, the first known superconstant lower bound
for this problem. For the upper bound, we use an approach inspired by
partial coloring methods for proving discrepancy upper bounds, and
obtain $(\eps, \delta)$-differentially private algorithms for range
counting with polynomially bounded shatter function range spaces.
\end{abstract}

\section{Introduction}

A \emph{range counting} problem is specified by a set $P$ of size $|P|
= n$, and a \emph{range space} $\range \subseteq 2^{P}$. Given a query
range $R \in \range$, the output is $|\{ p \in P \cap R\}|$.  More
generally, each point $p \in P$ has an integer {\em weight} $x_p$ and
the range returns $R(\vec{x}) = \sum_{p \in R}{x_p}$.  This problem is
fundamental in Computational Geometry and a workhorse in applications,
for various examples of range spaces from axis-parallel boxes
(orthogonal range counting), to regions bounded by hyperplanes
(halfspace counting) and beyond (e.g., simplices).  Orthogonal range
counting is commonly used in databases and data analysis.  Halfspace
counting is not only interesting in itself, but general algebraic
range counting can be ``lifted'' to a higher dimension and encoded as
halfspace counting~\cite{YaoY85}.

We study privacy of range counting. In private range counting the set
$P$ of points as well as the range space $\range$ are considered
public information, while the point weights $x_p$ are considered
private (and may denote, e.g.~number of users at a geographic
location). As the exact solution can reveal the private weights, we
need to turn to approximate solutions.  We define the \emph{average
  squared error} of an algorithm $\alg$ for range counting as
%on input $\vec{x}$  is 
%\begin{equation}
$  \frac{1}{|\range|}\sum_{R \in \range}{\left(\alg(R, \vec{x}) - R(\vec{x})\right)^2}$. 
%\end{equation}
For privacy, we adopt the well-established notion of differential
privacy.  A mechanism $\mech= \{M_n\}$ is \emph{$(\eps,
  \delta)$-differentially private} if for every $n$, every $\vec{x},
\vec{x}'$ with $\|\vec{x} - \vec{x}\|_1 \leq 1$, and every measurable $S
\subseteq \mathbb{R}^d$, the map $M_n$ satisfies
   \begin{equation}
     \Pr[M_n(\vec{x}) \in S] \leq e^\eps \Pr[M_n(\vec{x}') \in S] + \delta.
   \end{equation}

Surprisingly, very little is known about private range counting.
Applying methods of differential privacy from first principles
(Laplace noise and the basic composition theorem of differential
privacy) will add large --- variance $\Omega(n^2)$ in the case of
halfspace counting in the plane --- noise to each
output.  More generally, let $\vec{A}$ be an incidence matrix for a
range space $\range$ (i.e.~a matrix whose rows are the indicator
vectors of all ranges $R \in \range$) and let $\vec{x}$ be the weights.  The problem of
computing $\vec{A}\vec{x}$ is the range counting problem. The average
squared error of an approximate algorithm $\alg$ is
$\frac{1}{|\range|}\|\alg(\vec{x}) - \vec{A}\vec{x}\|_2^2$.  In
general, we can consider this problem for any $\vec{A} \in \{0,
1\}^{m\times n}$, not necessarily ones that correspond to natural ranges from
some constant dimensional geometric space. This is the {\em predicate counting} problem,
well-studied in differential privacy.  Then it is known that no
mechanism that has average squared error $o(n)$ can be
$(\epsilon,\delta)$-differentially private~\cite{Dinur2003,
  Dwork2007}. However, the lower bounds are obtained using random
$\vec{A}$'s that will not correspond to specific range spaces of interest. No
super-constant lower bounds are known against $(\eps, \delta)$-differential privacy
for natural problems like halfspace or orthogonal range
counting in constant dimensional space.
\footnote{Constant lower bounds follow from the work of
  Roth~\cite{Roth2010a} as well as from reductions from lower bounds
  for conjunction queries.}

\junk{ For
the stronger notion of $(\eps, 0)$-differential privacy, the best
known lower bound on \emph{worst-case squared error} in a restricted
setting~\cite{Roth2010a} is $\Omega(d^2)$, where $d$ is the
VC-dimension (or for non-integer $\vec{A}$, the fat-shattering
dimension) of the range space. There exist lower bound techniques
tailored to $(\eps, 0)$-differential privacy that separate this notion
from $(\eps, \delta)$-differential
privacy~\cite{Hardt:2010:GDP:1806689.1806786, De2011}.}

Our results are  for $(\eps, \delta)$-differentially private range counting, and 
use the combinatorial structure of $\vec{A}$'s for  range spaces. 
 Our main application is halfspace counting, but our approach is
 general and yields other results too.%, e.g.~improved range counting for constant VC-dimension range spaces.

 \medskip \noindent
$\bullet$ {\em (Halfspace counting upper bound)} 
The (primal) \emph{shatter function} of $\range$ is defined as
$\pi_\range(s) = \max_{X \in {P \choose s}} |\range|_X|$ (i.e.~the
number of \emph{distinct} sets in the restriction $\range|_X$).  The shatter function
of $\range$ defined by halfspaces in $d$-dimensions is bounded as $\pi_\range(s) =
O(s^d)$. 
%In general, any range space $\range$ of VC-dimension $d$ has
%$\pi_\range$ bounded as $O(s^d)$. 
%Thus the shatter function exponent
%is a finer-grained measure of complexity than VC-dimension, considered
%in prior work~\cite{Blum2008}.  

We show that there is an $(\eps, \delta)$-differentially private range
counting mechanism that achieves $O(n^{1 - 1/d})$ average squared
error for range spaces with shatter function bounded by $O(s^d)$, and
therefore for $d$-dimensional halfspace range counting.

Our upper bound shows that previous lower bounds~\cite{Dinur2003,
  Dwork2007} for general $\vec{A}$'s indeed do {\em not} apply to
halfspace range counting.  Our algorithm runs in time polynomial in
$n$ and $m$.  Previous work on this problem is incomparable.  Work by
Blum, Ligett and Roth~\cite{Blum2008} gave a non-constructive squared
error upper bound of $O(d^2n^{4/3})$ for range spaces with
VC-dimension $d$ and a matching constructive bound for halfspace range
counting for $(\eps, 0)$-differential privacy with a slightly
different objective.  Since the shatter function of a range space with
VC-dimension $d$ is bounded by $O(s^d)$, our result also implies a
constructive approximation upper bound of $O(n^{1 - 1/d})$ for
VC-dimension $d$ range spaces.  \junk{ although using a slightly
  different and not directly comparable notion of
  utility. Furthermore, their upper bound satisfies the stronger
  notion of $(\eps, 0)$-differential privacy and applies to the large
  universe setting where $\|\vec{x}\|_1 \leq n$ and $|P| = N$, $N \gg
  n$ ($n$ and $N$ being public parameters).  }

Our approach relies on prior work~\cite{matousek1995tight} to
decompose the range space into a logarithmic number of range spaces,
some of them consisting only of small ranges, and some containing a
small number of distinct ranges. We exploit this trade-off between
maximum range size and number of distinct ranges by combining
randomized response and Laplacian noise based differentially private
mechanisms, but this balancing still leaves us with large noise in
some cases. Nevertheless, we can bound the average privacy loss over
the points $p \in P$. Our main idea is to use this approach to
preserve privacy for {\em most} points $p \in P$; the shatter function
bound does not increase for restrictions of $P$ and $\range$ and we
can recurse on the remaining points of $P$. This argument is inspired
by partial coloring methods used in discrepancy theory\junk{ and may be
applicable to other problems in differential privacy}. \hfill
$\blacksquare$

\medskip \noindent $\bullet$ {\em (Range counting lower bound)} For
halfspace counting in $d$ dimensions, we show that any mechanism that
has average squared error within $o(n^{1-1/d})$ is not
$(\eps,\delta)$-differentially private for any constant $\eps$ and
$\delta$.  \junk{Hence, for small $d$ ($2,3,..)$, the lower bound is
  polynomially smaller than the general lower bound.} We prove this
lower bound using a notion of {\em discrepancy} where, in contrast to
the standard notion where $\{+1,-1\}$ colorings are considered, we allow
$\{0,+1,-1\}$ colorings but subject to some budget constraints on
$\{+1, -1\}$. The budget constraints allows us to relate this notion
of discrepancy to the classical one. Once the approach via the correct
notion of discrepancy is developed, the mechanics are simple. Lower
bounds will follow from combinatorial analysis of the discrepancy of
range spaces. For orthogonal range counting, our approach immediately
gives a lower bound of $(\log n)^{d - O(1)}$ on the average squared
error of any $(\eps, \delta)$ differentially private mechanism.  The
best upper bound in this setting is the work of Chan, Shi, and
Song~\cite{chan2010private} who give an algorithm with average squared
error $O((\log n)^{2d})$. No previous super-constant lower bounds are
known for this problem even for large constant $d$. We note that
proving a tight lower bound on the combinatorial discrepancy of
axis-aligned boxes in $d$ dimensions is a major open problem in
discrepancy theory, and any improvement to the current discrepancy
lower bound will yield a corresponding improvement in lower bounds for
privacy. 
\hfill $\blacksquare$

%\todo{Check the upper.lower bound claims above. Add some general discussion abt relation to prior work? Add discussion of  $\eps, \delta$ vs $\eps$ differential privacy., Is our upper bound constrctive? Change hyperplane throughout to halfspace where applicable? 
%Does Blum et al give $d n^2/3$ upper bound for our problem (although in a different model)? } 
   
In Section~\ref{sect:prior} we review related prior work. In
Section~\ref{sect:defn}, we define concepts we need, including
differential privacy and suitable notions of discrepancy. In
Section~\ref{sect:lb}, we present our lower bounds, and in
Section~\ref{sect:ub}, the upper bounds.  We describe extensions and
alternative algorithmic solutions in Section~\ref{sect:extensions}. 

\section{Prior Work}
\label{sect:prior}

There is a rich and growing literature on solving counting problems
while satisfying strong privacy guarantees. We will survey the prior
work that is most relevant to our results.

In a seminal paper, Dinur and Nissim~\cite{Dinur2003} initiated the
study of the limits of output perturbation in answering arbitrary
counting queries privately. They showed that if an algorithm $\alg$
satisfies $\|\alg(\vec{x}) - \vec{A}\vec{x}\|_\infty^2 = o(n)$ for a
random 0-1 matrix $\vec{A}$, then an adversary can reconstruct
$\vec{x}$ almost exactly, implying that the algorithm is not $(\eps,
\delta)$-differentially private for any constant $\eps,
\delta$.\footnote{Our methods based on discrepancy allow us to
  re-prove the lower bound of Dinur and Nissim, as well as the version
  of Dwork and Yekhanin~\cite{Dwork2008} that uses an explicit
  $\vec{A}$.} There is relatively little prior work on negative
results for $(\epsilon, \delta)$-differential privacy for natural
restrictions of $\vec{A}$.  An exception is the work on lower bounding
the noise necessary to privately answer conjunction
queries~\cite{Kasiviswanathan2010, De2011}. Conjunction queries on a
database with $d$ attributes can be reduced to answering orthogonal
range counting or halfspace range counting queries in $d$
dimensions. When $d$ is constant, the lower bounds on conjunction
queries imply a lower bound of $C^d$ (for an absolute constant $C>1$)
on the average squared error neccessary to answer $d$-dimensional
halfspace or orthogonal queries privately (here and in the remainder
of this section we suppress dependence on $\eps$, $\delta$, and the
probability of failure).  In other related work, Roth~\cite{Roth2010a}
showed that linear queries with fat shattering dimension $D$ require
squared noise $\Omega(D^2)$ to preserve privacy. The fat shattering
dimension reduces to the VC-dimension for counting queries, and has
value $d+1$ for the range space of halfspaces in $d$ dimensions.  No
super-constant lower bounds were previously known for $(\eps,
\delta)$-differential privacy for the halfspace range counting or
orthogonal range counting problems in constant dimensional space.

\junk{ Their
results were subsequently strengthened in several directions: lower
bounds against algorithms that incur unbounded error on a constant
fraction of the queries~\cite{Dwork2007} (which imply lower bounds for
average squared error as well), and lower bounds for explicit
$\vec{A}$~\cite{Dwork2008}. We note that the lower bound for explicit
$\vec{A}$ of~\cite{Dwork2008} can be proven using our methods. There
is relatively little prior work on negative results for$(\epsilon,
\delta)$-differential privacy for natural restrictions of
$\vec{A}$. }

\junk{ \cite{Dwork2007} extended the result to show that even
algorithms that add arbitrary noise to a constant fraction of the
queries and small noise to the rest are blatantly non-private. Dwork
and Yekhanin~\cite{Dwork2008} improved the running time of the
adversary and showed that there exists an explicit set of queries,
coming from the Fourier transform on cyclic groups, that suffices for
the reconstruction attack. In recent work~\cite{De2011}, De showed
that when the database size is larger than the universe size, even
larger noise is required to prevent blatant non-privacy. We note that
the lower bound of Dinur and Nissim as well as the the result of Dwork
and Yekhanin that Fourier queries suffice for reconstruction can be
reproved using our methods and standard results in discrepancy
theory. While we do not emphasize efficiency in the presentation of
our reconstruction arguments (as an inefficient attack suffices to
contradict $(\eps, \delta)$-differential privacy), most of our
reconstruction attacks can be made efficient using the linear
programming approach of Dinur and Nissim and an appropriate rounding
procedure.}

The study of private range counting for restricted range spaces was
initiated with the work of Blum, Ligett, and Roth~\cite{Blum2008},
who, using an argument based on epsilon nets, showed that queries of
VC dimension $d$ can be answered with worst-case squared noise
$O(d^2n^{4/3})$. Their algorithm is not computationally efficient, but
they gave efficient algorithms with comparable guarantees for the
interval range counting and halfspace range counting
problems. Although their error bound is inferior to ours (when the
size of the database is comparable to the universe size),
the models are not directly comparable. While we consider a
finite universe, they consider a continuous space, but give relaxed
utility guarantes, namely that each query answer is accurate for a
halfspace close to the query halfspace. Additionally, their
algorithms satisfy the stronger notion of $(\eps, 0)$-differential
privacy and accomodate the regime where $\|\vec{x}\|_1$ is public and
bounded by $n$ and $P$ is much larger.  

For interval queries, the work of Blum, Ligett, and Roth was
subsequently improved by Xiao, Wang, and Gehrke~\cite{Xiao2010}
(in the regime where database size and universe size are
comparable), who gave a polylogarithmic noise upper bound via the
wavelet transform. A related algorithm that achieves an average
squared error upper bound of $O((\log n)^{3d})$ for $d$-dimensional
orthogonal range counting was given by Chan, Shi, and
Song~\cite{chan2010private}. We note that if we relax the privacy
guarantee of Chan, Shi, and Song to $(\eps, \delta)$-differential
privacy, their algorithm can be analyzed to provide average squared
error $O(\log^{2d}n)$.

Much subsequent work has focused on answering $m$ arbitrary queries
efficiently with squared error linear in $n$ and polylogarithmic in
$m$~\cite{Dwork:2009:CDP:1536414.1536467, Roth2010, Hardt2010,
  Hardt2010a, Gupta2011a}. A related line of work investigates the
problem of answering conjunction queries with optimal error~\cite{Gupta2011,
  Barak:2007:PAC:1265530.1265569}. 

\vspace{+3pt}
\noindent\textbf{Prior work for $(\eps, 0)$-differential privacy}. Stronger lower bounds can be shown when $\delta = 0$, and
there are known separations between the cases $\delta = 0$ and $\delta
> 0$, even when $\delta$ is superpolynomially
small~\cite{De2011}. Hardt and
Tulwar~\cite{Hardt:2010:GDP:1806689.1806786} gave a lower bound for
linear queries based on geometric properties of the query matrix
$\vec{A}$. De~\cite{De2011} simplified and extended their lower bound
results.\junk{ He also extended it to show a separation between
  counting queries and arbitrary low-sensitiviy queries. } Blum,
Ligett, and Roth~\cite{Blum2008} showed that no $(\eps,
0)$-differentially private mechanism can answer interval queries
with any nontrivial noise when the universe is continuous.
\junk{Also, as our lower bounds go to infinity as the universe size
  $n$ increases, they imply the lower bound of~\cite{Blum2008} for
  halfspace counting in dimension $d> 1$ for the weaker setting
  $\delta>0$.}

\junk{ The problem can take several variants. In the
non-interactive setting, an algorithm can produce a synopsis of the
database, as well as an actual synthetic database, the latter task
known as data sanitization. The work of Blum, Ligett and Roth actually
proved the existence of a data sanitization mechanism with the
guarantees given above. The computational complexity of data
sanitization when the universe or query class is very large was
investigated in~; they proved
that under standard assumptions data sanitization is hard and provided
a more efficient sanitization algorithm with noise guarantees
incomparable with these of Blum, Ligett, and Roth. Roth and
Roughgarden~\cite{} extended the results of Blum, Ligett, and
Roth to the interactive setting where the queries are not known in
advance. An algorithm based on the multiplicative weights method was
proposed by Hardt and Rothblum~\cite{Hardt2010a}; their algorithm
provided tighter noise guarantees (on the order of $\sqrt{n}$ rather
than $n^{2/3}$ and better running time. Subsequent work has simplified
and generalized the results of Hardt and Rothblum~\cite{Hardt2010a,
  Gupta2011a, Gupta2011a}. \cite{Gupta2011a} gave a new analysis of the
median mechanism of Roth and Roughgarden based on the techniques of
Hardt and Routhblum; the analysis showed that the median mechanism can
achieve noise comparable to that of the multiplicative weights
mechanism. }

\vspace{+3pt}
\noindent\textbf{Discrepancy theory}. For background in discrepancy theory we refer the reader to the books
of Chazelle~\cite{chazelle2000discrepancy} and
Matou\u{s}ek~\cite{matousek2010geometric}. Chazelle provides an overview
of the applications of discrepancy theory to computer science, while
Matou\u{s}ek gives a survey of discrepancy theory results for geometric
range spaces. 

\vspace{+3pt}
\noindent\textbf{Geometric range counting}. Geometric range counting
and the closely related problems of range sums and range searching have a rich history in computational geometry. We
refer the reader to the survey of Agarwal and
Erickson~\cite{agarwal1999geometric} for background.

\section{Preliminaries}
\label{sect:defn}

We typeset vectors and matrices as $ \vec{x}$, $\vec{A}$ and their
elements as $x_j$, $A_{ij}$. We denote the $i$-th row of $\vec{A}$ as
$\vec{A}_{i*}$ and the $j$-th column as $\vec{A}_{*j}$. Given a matrix
$\vec{A}$, the function $\col(\vec{A})$ equals the number of columns
of $\vec{A}$. For a matrix
$\vec{A}$ with $n$ columns, and a set $S \subseteq [n]$ we use
$\vec{A}|_S$ to denote the submatrix of $\vec{A}$ consisting of the
columns corresponding to elements of $S$ (with duplicated rows
removed). Similarly, for a range space $\range$ with incidence matrix
$\vec{A}$, the range space $\range|_S$ is the one corresponding to the
incidence matrix $\vec{A}|_S$. We denote the $i$-th standard basis
vector $(0, \ldots, 0, 1, 0, \ldots, 0)^T$ (where $1$ is in the $i$-th
coordinate) as $\vec{e_i}$. For a set $P$ we denote the collection of
subsets of $P$ of size $s$ as ${P \choose s}$. 

\subsection{Range Counting}

We will use the definitions  for range counting, average
squared error, orthogonal and hyperspace range counting, as well as
the linear algebraic notation introduced in the Introduction. We also consider
worst-case squared error, which for an algorithm $\alg$ and a range
space with incidence matrix $\vec{A}$ is $\|\alg(\vec{x}) -
\vec{A}\vec{x}\|_\infty^2 \geq \frac{1}{m} \|\alg(\vec{x}) -
\vec{A}\vec{x}\|_2^2$. We give all our lower bounds in average squared
error and state our upper bounds in terms of both average and
worst-case squared error.

\junk{A \emph{range counting} problem is specified by a set $P$ of size
$|P| = n$, and a \emph{range space} $\range \subseteq 2^{P}$. An exact
range counting algorithm, on input $\vec{x} \in \mathbb{Z}^P$ and a range $R
\in \mathcal{R}$, outputs
\begin{equation}
  R(\vec{x}) = \sum_{p \in R}{x_p}.
\end{equation}

We define the \emph{average squared error} of an algorithm $\alg$ on
input $\vec{x}$ \junk{for a range counting problem} as
\begin{equation}
  \frac{1}{|\range|}\sum_{R \in \range}{\left(\alg(R, \vec{x}) - R(\vec{x})\right)^2}
\end{equation}
We define the worst case squares error of $\alg$ as
\begin{equation}
  \max_{R \in \range}\left\{(\alg(R, \vec{x}) -  R(\vec{x}))^2\right\}
\end{equation}

The concepts of range counting, average squared error, and maximum
error can be concisely formulated using matrix notation.  Let
$\vec{A}$ be the incidence matrix for a range $\range$. Then on input
$\vec{x}$, an exact algorithm outputs $\vec{A}\vec{x}$. The average
squared error of an approximate algorithm $\alg$ is $\frac{1}{|\range|}\|\alg(\vec{x}) -
\vec{A}\vec{x}\|_2^2$, and the worst-case error is $\|\alg(\vec{x}) -
\vec{A}\vec{x}\|_\infty$. 

We will consider the following range counting problems:
\begin{itemize}
\item in the \emph{hyperplane range counting} problem  in dimension
  $d$, $P$ is a set of points in $\mathbb{R}^d$, and $\range$ is a
  range space of subsets of $P$ that can be written as $H
  \cap P$ for some halfspace $H$ in $\mathbb{R}^d$.
\item  in the \emph{orthogonal range counting} problem in dimension
  $d$, $P$ is a set of points in $\mathbb{R}^d$, and $\range$ is a
  range space of  subsets of $P$ that can be written as $B
  \cap P$ for some axis-parallel box $B$ in $\mathbb{R}^d$.
\end{itemize}
}%end of junk

The \emph{VC-dimension} of a range space $\range$ is defined as the size of
the  largest set $X \subseteq P$ such that $\range|_X =
2^X$. The (primal) \emph{shatter function} of $\range$ is defined as
$\pi_\range(s) = \max_{X \in {P \choose s}} |\range|_X|$ (i.e.~the
number of \emph{distinct} sets in $\range|_X$). 

\begin{fact}[\cite{matousek2010geometric}]
  If the VC-dimension of $\range$ is $d$, then $\pi_\range(s) =
  O(s^d)$. Conversely, if $\pi_\range(s) = s^{O(1)}$ then the
  VC-dimension of $\range$ is constant.
\end{fact}

\begin{fact}[\cite{matousek2010geometric}]
  The VC-dimension of the range space $\range$ induced on $P$ by all
  halfspaces in $\mathbb{R}^d$ is $d+1$. The shatter function of
  $\range$ is bounded as $\pi_\range = O(s^d)$. 
\end{fact}

\junk{In general, a range space that can be defined by a parametrized
constant-degree polynomial inequality with $d$ parameters has shatter
function bounded as $O(s^d)$. }

\subsection{Differential Privacy}

%\begin{definition}
  For any two sets $\uni$ (\emph{the universe}) and $Y$, a
  \emph{mechanism} $\mech$ over $\uni$ with range $Y$ is a family
  of maps $\{M_n\}$, $M_n: \uni^n \rightarrow \rho(Y)$, where
  $\rho(Y)$ is the set of random variables that take values in
  $Y$.
%\end{definition}
For the rest of this paper, we will focus on mechanisms over
$\mathbb{Z}$ or over $\{0, 1\}$, with range $\mathbb{R}^m$.

\begin{definition}
  A mechanism $\mech = \{M_n\}$ over (a subset of) $\mathbb{Z}$ with
  range $Y$ is \emph{$(\eps, \delta)$-differentially
    private} if for every $n$, every $\vec{x}, \vec{x}'$ with
  $\|\vec{x} - \vec{x}'\|_1 \leq 1$, and every measurable $S \subseteq
  Y$, the map $M_n$ satisfies
   \begin{equation*}
     \Pr[M_n(\vec{x}) \in S] \leq e^\eps \Pr[M_n(\vec{x}') \in S] + \delta.
   \end{equation*}
\end{definition}

For lower bounds we use the following claim, which implies that being able
to decode most of the input from the output contradicts differential privacy.
\junk{
\begin{definition}
  A mechanism $\mech=\{M_n\}$ over the universe $\mathbb{Z}$ is \emph{$(\alpha,
  \beta)$-non-private} if for some $n$ there exists a (not necessarily efficient)
  algorithm $\alg$ such that
  \begin{equation}
    \forall \vec{x} \in \mathbb{Z}^n: \Pr[\|\alg(M_n(\vec{x})) -
    \vec{x}\|_1 > \alpha n] < \beta.
  \end{equation}
\end{definition}
}

\begin{lemma}[\cite{De2011}]
  \label{lm:non-priv-dp}
  Let $\mech = \{M_n\}$ be a mechanism such that for some $n$ there
  exists a (not necessarily efficient) algorithm $\alg$ such that
  \begin{equation*}
    \forall \vec{x} \in \mathbb{Z}^n: \Pr[\|\alg(M_n(\vec{x})) -
    \vec{x}\|_1 >  \alpha n ] < \beta.
  \end{equation*}
  Then there exist $\eps = \eps(\alpha, \beta)$ and $\delta =
  \delta(\alpha, \beta)$ such that the mechanism $\mech$ is not $(\eps,
  \delta)$-differentially private.
\end{lemma}

\junk{%I can fill in a proof, too
\begin{definition}
  Let $\mech^1 = \{M^1_n\}, \ldots, \mech^s = \{M^s_n\}$ be mechanisms
  over $\mathbb{Z}$. Their \emph{composition} is the mechanism that on
  input $\mathbb{Z}^n$ outputs $(M^1_n(\vec{x}), \ldots,
  M^s_n(\vec{x}))$. We call the composition a \emph{degree-$t$
    composition of $(\eps, \delta)$-differentially private
    mechanisms}, if $\mech^k$ is $(\eps, \delta)$-differentially
  private for every $k \in [s]$ conditional on the output of $\mech^1,
  \ldots, \mech^{k-1}$, and, moreover, for all $n$ and all $i \in
  [n]$, for any $\vec{x}, \vec{x}' \in \mathbb{Z}^n$ such that
  $\supp(\vec{x} - \vec{x}') = \{i\}$, we have
  \begin{equation}
    M^k_n(\vec{x}) \equiv M^k_n(\vec{x}')
  \end{equation}
  for at least $s - t$ values of $k$. 
\end{definition}

\begin{lemma}
  \label{lm:composition}
  Let $\mech$ be a degree-$t$ composition of $(\eps, 0)$
  differentially private mechanisms. Then $\mech$ satisfies
  $(\sqrt{2t\ln(1/\delta)}\eps + t\eps(e^\eps - 1),
  \delta)$-differential privacy for any $\delta > 0$.
\end{lemma}
Notice that for sufficiently small $\eps$ and $t = \Omega(\eps^{-2})$,
a degree-$t$ composition satisfies $(O(\sqrt{t\log(1/\delta)}\eps,
\delta)$-differential privacy.
 }

A basic mechanism to achieve differential privacy with $\delta = 0$ is
the \emph{Laplace noise mechanism}, first proposed in~\cite{DMNS}. Let
us here and for the rest of the paper denote by $\Lap(s)$ the Laplace
distribution centered at 0 with scale parameter $s$. 
\begin{lemma}[\cite{DMNS}]
  \label{lm:laplace}
  Let $f$ be any real-valued function which for any $\vec{x}, \vec{x}'
  \in \mathbb{Z}^n$ such that $\|\vec{x} - \vec{x}'\|_1 \leq 1$
  satisfies $|f(\vec{x}) - f(\vec{x}')| \leq 1$. Then the mechanism
  that on input $\vec{x}$ outputs $f(\vec{x}) + \Lap(1/\eps)$
  satisfies $(\eps, 0)$-differential privacy.
\end{lemma}

%\begin{definition}
  The \emph{composition} of mechanisms $\mech^1 = \{M^1_n\}$, $\ldots$,
  $\mech^s = \{M^s_n\}$ is the mechanism that on input $\mathbb{Z}^n$
  outputs $(M^1_n(\vec{x})$, $\ldots$, $M^s_n(\vec{x}))$.
%\end{definition}
We need the following composition lemma first proved in~\cite{DMNS}. 
\begin{lemma}[\cite{DMNS}]
  \label{lm:simple_comp}
  Let the mechanisms $\mech^1, \ldots, \mech^s$ satisfy, respectively,
  $(\eps_1, \delta_1), \ldots, (\eps_s, \delta_s)$ differential
  privacy. The composition $\mech$ of the mechanisms satisfies
  $(\sum_i{\eps_i}, \sum_i{\delta_i})$-differential privacy.
\end{lemma}

We also need a stronger result, which is a straightforward
extension of the composition theorem of Dwork, Rothblum, and
Vadhan~\cite{Dwork2010}. 
To state the result we define a notion of
privacy loss. Following~\cite{Dwork2010}, let us first define the
\emph{maximum divergence} of two random variables $a$ and $b$ as 
\begin{equation*}
  D_\infty(a \| b) = \max_{S} \ln \frac{\Pr[a \in S]}{\Pr[b \in S]},
\end{equation*}
where $S$ ranges over measurable subsets of the support of $b$. Note
that a mechanism $\mech = \{M_n\}$ is $(\eps, 0)$-differentially
private if and only if for every $n$ and any $\vec{x}, \vec{x}':
\|\vec{x} - \vec{x}'\|_1 \leq 1$, we have $D_\infty(M_n(\vec{x}) \|
M_n(\vec{x}')) \leq \eps$ and $D_\infty(M_n(\vec{x}') \| M_n(\vec{x}))
\leq \eps$. 

\begin{definition}
  Let $\mech$ be a composition of $\mech^1, \ldots, \mech^s$. The
  \emph{privacy loss of $i \in [n]$ for the $j$-th output}  is
  \begin{equation*}
    l_\mech(i, j) = \max_{\vec{x}, \vec{x}' = \vec{x} \pm
      \vec{e_i}}{D_\infty(M^j_n(\vec{x}) \| M_n^j(\vec{x}'))},
  \end{equation*}
  The \emph{($\ell_2$) privacy loss of $i \in [n]$} is $L_\mech(i) =
  \sqrt{\sum_{j \in [s]}{l_\mech(i, j)^2}}$.
\end{definition}

\begin{lemma}
  \label{lm:composition}
  Let $\mech$ be a composition of $\mech^1, \ldots, \mech^s$ and let
  $\eps > \max_{i \in [n]}{L_\mech(i)}$. Then, for any $\delta > 0$,
  $\mech$ satisfies $(\sqrt{2\ln(1/\delta)}\eps, \delta)$-differential
  privacy.
\end{lemma}

Note that for the range counting problem, the privacy loss is defined
for a point $p$.

\subsection{Discrepancy}

Here we define a modified notion of discrepancy. In
Section~\ref{sect:lb}, we show that this modified notion of
discrepancy is useful in carrying out Dinur-Nissm type attacks on
privacy.

\begin{definition}
  For any $\vec{A} \in \mathbb{R}^{m \times n}$, we define
  \begin{align*}
    \disc_{p, \alpha}(\vec{A}) &= \min_{\substack{\vec{x} \in \{0, \pm 1\}^n \\
        \|\vec{x}\|_1 \geq  \alpha \col(A) }} {\|\vec{A}\vec{x}\|_p}\\
    \herdisc_{p, \alpha}(\vec{A}) &=  
    \max_{S \subseteq [n]}{\disc_{p, \alpha}(\vec{A}|_S)}.
  \end{align*}
\end{definition}

The standard notions of discrepancy and hereditary discrepancy
correspond to the special cases $\disc = \disc_{\infty, 1}$ and
$\herdisc = \herdisc_{\infty, 1}$. The cases $\disc_{2, 1}$ and
$\herdisc_{2, 1}$ have also been extensively studied, especially as
means of proving lower bounds on $\disc$ and $\herdisc$. On the other
hand the case $\disc_{p, 0}$ is trivially the identically 0 function.
Next, we exhibit a connection between $\herdisc_{p, 1}$ and $\herdisc_{p,
  \alpha}$ for $\alpha \in (0, 1)$ and any $p$.  \junk{\begin{claim}
  \begin{equation}
   \herdisc_{p, 1} \leq \frac{\log n}{\log 1/(1-\alpha)}\herdisc_{p, \alpha}.     
  \end{equation}
\end{claim}
However, when discrepancy can be expressed as a polynomial function of
the number of elements, we can get tighter relationships between
$\herdisc_{p, 1}$ and $\herdisc_{p, \alpha}$.}
\begin{lemma}
  \label{cl:std2alpha-disc}
  Let $f(s) = \max_{S\subseteq [n]: |S| \leq s}{\disc_{p,
      \alpha}(\vec{A}|_S)}$. Then $\disc_{p, 1}(\vec{A}) \leq \sum_{i
    = 0}^\infty{f((1 - \alpha)^in)}$, and, therefore,
  $\herdisc_{p, 1}(\vec{A})$ $\leq {\sum_{i = 0}^\infty{f((1 -
      \alpha)^in)}}$
\end{lemma}
\begin{proof}
  We will
  find an assignment $\vec{x} \in \{\pm 1\}^n$ such that 
  $\|\vec{A} \vec{x}\|_p \leq \sum_{i   = 0}^\infty{f((1 -
    \alpha)^in)}$, which is sufficient to prove the lemma.
  Let $\vec{x}' \in \{0, \pm 1\}^n$ be such that $\|\vec{A}
  \vec{x}\|_p \leq f(n)$ and $\|\vec{x}\|_1 \geq \alpha n$. Let $S
  = \{i: x_i = 0\}$. Since $\|\vec{x}\|_1 \geq \alpha n$, $|S| \leq
  (1-\alpha)n$. We recurse to find an assignment $\vec{x}'' \in \{\pm
  1\}^S$ such that $\|(\vec{A}|_S) \vec{x}''\|_p \leq \sum_{i = 0}^\infty{f((1
    - \alpha)^i|S|)} \leq \sum_{i = 1}^\infty{f((1 -
    \alpha)^i n)}$. Set $x_i = x'_i$ when $i \not \in S$ and $x_i =
  x''_i$ when $i \in S$.
\end{proof}

Lemma~\ref{cl:std2alpha-disc} and the observation
$\herdisc_{p, \alpha} = \max_{s = 1}^n f(s)$ imply that for any
$\vec{A}$, 
\begin{equation*}
\herdisc_{p, 1}(\vec{A}) \leq \frac{\log n}{\log
  1/(1-\alpha)} \herdisc_{p, \alpha}(\vec{A}). 
\end{equation*}
However using Lemma~\ref{cl:std2alpha-disc} directly and the
observation that a restriction of a halfspace range space (or a range
space of axis-aligned boxes) is a range space of the same kind, we get
stronger lowerbounds for $\herdisc_{p, \alpha}$. Below we list several
interesting results that can be derived in this way from known results
in combinatorial discrepancy theory~\cite{chazelle2000discrepancy,
  matousek2010geometric}. Below we provide more specific references to
the discrepancy lower bound used to derive each result. We provide a
full proof of the first result; the remaining proofs follow analogous
reasoning.

\begin{lemma}[\cite{chazelle1995elementary}]
  \label{lm:hyperplanes-disc}
   For infinitely many $n$ there exists a set of $n$ points $P$ and
  $m$ halfspaces $H_1, \ldots, H_m$ in $\mathbb{R}^d$ ($d = O(1)$)
  such that the following holds. Let $\vec{A}$ denote the incidence
  matrix of the collection of sets $\{H_j \cap P, j \in [m]\}$. Then
  for any $\alpha = \Omega(1)$,
   $\herdisc_{2, \alpha}(\vec{A}) = \Omega(m^{1/2}n^{1/2 - 1/2d}).$
\end{lemma}
\begin{proof}
  Assume for contradiction that all but finately many $m \times n$
  incidence matrices $\vec{A}$ of halfspaces in $\mathbb{R}^d$ have
  hereditary $\alpha$-discrepancy $\herdisc_{2, \alpha}(\vec{A}) =
  o(m^{1/2}n^{1/2 - 1/2d})$.  By the results in
  \cite{chazelle1995elementary}, there exist infinitely many sets of
  $n$ points $P$ and $m = {n \choose d}$ halfspaces $H_1, \ldots, H_m$
  such that the incidence matrix $\vec{B}$ of $\{H_j \cap P, j \in
  [m]\}$ has hereditary discrepancy $\herdisc_{2, 1} =
  \Omega(m^{1/2}n^{1/2 - 1/2d})$. Let us fix any such set of points
  and halfspaces and the corresponding incidence matrix $\vec{B}$. Any
  restriction $\vec{B}|_S$ for $S \subseteq P$ is also the incidence
  matrix of sets induced by points and halfspaces, and by assumption,
  $\herdisc(\vec{B}|_S) = o(m^{1/2}|S|^{1/2 - 1/2d})$. Plugging this
  bound in Lemma~\ref{cl:std2alpha-disc} we get $\herdisc_{2,
    1}(\vec{B}) = o(m^{1/2}n^{1/2 - 1/2d})$, a contradiction.
\end{proof}

\begin{lemma}[\cite{roth1954irregularities,Beck:fk}]
  \label{lm:boxes-disc} 
  For infinitely many $n$ there exists a set of $n$ points $P$ and $m$ axis-parallel boxes $B_1,
  \ldots, B_m$ in $\mathbb{R}^d$ ($d = O(1)$) such that the following holds. Let
  $\vec{A}$ denote the incidence matrix of the collection of sets $\{B_j
  \cap P, j \in [m]\}$. Then for any $\alpha = \Omega(1)$,
 %   \begin{equation}
    $  \herdisc_{2, \alpha}(\vec{A}) = \Omega(m^{1/2}(\log n)^{d/2 - 3/2}).$
  %  \end{equation}
\end{lemma}

\begin{lemma}[\cite{chazelle2001trace}]
  \label{lm:boxes-highdim-disc}
  For infinitely many $n$ there exists a set of $n$ points $P$ and
  $m$ axis-parallel boxes $B_1, \ldots, B_m$ in $\mathbb{R}^d$ ($d =
  \Theta(\log n)$) such that the following holds. Let $\vec{A}$ denote
  the incidence matrix of the collection of sets $\{B_i \cap P, j \in
  [m]\}$. Then for any $\alpha = \Omega(1)$,
 %   \begin{equation}
   $   \herdisc_{\infty, \alpha}(\vec{A}) = n^{\Omega(1)}.$
  %  \end{equation}
\end{lemma}

\begin{lemma}[\cite{spencer1985six}]
  \label{lm:gen-disc}
  For any $n$ and $m > n$ there exists a matrix $\vec{A} \in \{0,
  1\}^{m \times n}$ such that $\herdisc_{\infty, \alpha}(A) =
  \Omega(\sqrt{n \log 2m/n})$.
\end{lemma}

\section{Lower Bounds for Privacy from Discrepancy}
\label{sect:lb}

Our main result in this section is a noise lower bound on $(\eps,
\delta)$-differentially private mechanisms that approximate range
counting queries for a host of natural geometric range spaces. Our
main conceptual contribution is in identifying $\herdisc_{p, \alpha}$
as the key quantity in showing lower bounds against $(\eps,
\delta)$-differential privacy via a Dinur-Nissim type attack, and
connecting this quantity to the standard notion of combinatorial
discrepancy.

\begin{theorem}
  \label{thm:main-lb}
  For any $\alpha, \beta$, there exist $\eps(\alpha, \beta)$ and
  $\delta(\alpha, \beta)$ such that no mechanism $\mech = \{M_n\}$
over the universe $\{0, 1\}$ with range $\mathbb{R}^m$ that \junk{for some
$n$} for some $p$ satisfies
\begin{equation*}
  \forall \vec{x} \in \{0, 1\}^n: \Pr[\|M_n(\vec{x}) -
  \vec{A}\vec{x}\|_p < \disc_{p, \alpha}(\vec{A})/2] \geq 1- \beta,
\end{equation*}
is $(\eps, \delta)$-differentially private.
\end{theorem}

We extend the lower bound to $\herdisc_{p, \alpha}$. This allows us to
use the connection between $\herdisc_{p, \alpha}$ and standard
discrepancy.
\begin{cor}
  \label{cor:mainlb}
  For any $\alpha, \beta$, there exist $\eps(\alpha, \beta)$ and
  $\delta(\alpha, \beta)$ such that no mechanism $\mech = \{M_n\}$
  over the universe $\{0, 1\}$ with range $\mathbb{R}^m$ that \junk{for some
  $n$} for some $p$  satisfies
\begin{equation*}
  \forall \vec{x} \in \{0, 1\}^n: \Pr[\|M_n(\vec{x}) -
  \vec{A}\vec{x}\|_p < \herdisc_{p, \alpha}(\vec{A})/2] \geq 1- \beta,
\end{equation*}
is $(\eps, \delta)$-differentially private.
\end{cor}
\begin{proof}
  We claim that given $M_n$ and any set $S \subseteq [n]$, we can
  construct $M'_n$ that takes as input $\vec{x}|_S$, is
  $(\eps, \delta)$-differentially private (with respect to
  $\vec{x}|_S$), and satisfies
  \begin{equation*}
    \forall \vec{x}|_S: \Pr[\|M'_n(\vec{x}|_S) -  (\vec{A}|_S)(\vec{x}|_S)\|_p
    < \herdisc_{p, \alpha}(\vec{A}|_S)/2] \geq 1- \beta. 
  \end{equation*}
  Then we can take $S$ such that $\disc_{p, \alpha}(\vec{A}|_S) =
  \herdisc_{p, \alpha}(\vec{A})$, and the corollary follows from
  Theorem~\ref{thm:main-lb}. 

  We define $M_n'$ as follows: $M_n'(\vec{x}|_S)$ extends
  $\vec{x}|_S$ to $\vec{x}$ by setting $x_i = 0$ for all $i \not \in
  S$ and outputs $M_n(\vec{x})$. It's easy to verify that $M_n'$
  satisfies the claimed properties.
\end{proof}

Theorem \ref{thm:main-lb} follows from Lemma~\ref{lm:non-priv-dp} and
the following lemma.
\begin{lemma}
  There exists a deterministic (not necessarily efficient) algorithm
  $\alg$ that on input a matrix $\vec{A} \in \mathbb{R}^{m \times n}$
  and a vector $\tilde{\vec{y}} \in \mathbb{R}^m$ satisfying
  $\|\tilde{\vec{y}} - \vec{A}\vec{x}\|_p < \disc_{p,
    \alpha}(\vec{A})/2$ for some $\vec{x} \in \{0, 1\}^n$, outputs a
  vector $\vec{x}' \in \{0, 1\}^n$ such that $\|\vec{x}' - \vec{x}\|_1
  \leq \alpha n$.
\end{lemma}
\begin{proof}
  Given $\tilde{\vec{y}}$, $\alg$ outputs an arbitrary $\vec{x}' \in
  \{0, 1\}^n$ such that $\|\vec{A}\vec{x}' - \tilde{\vec{y}}\|_p <
  \disc_{p, \alpha}(\vec{A})/2$. Such a $\vec{x}'$ exists, since
  $\|\vec{A}\vec{x} - \tilde{\vec{y}}\|_p < \disc_{p,
    \alpha}(\vec{A})/2$ by assumption.  We claim that $\|\vec{x} -
  \vec{x}'\| \leq \alpha n$. For contradiction, assume $\|\vec{x} -
  \vec{x}'\| > \alpha n$. Notice that $\vec{x} - \vec{x}' \in \{0, \pm
  1\}$. Then, by the definition of $\disc_{p, \alpha}$, $\|\vec{A}(\vec{x} - \vec{x}')\|_p \geq \disc_{p, \alpha}(\vec{A})$.
  By the triangle inequality, the assumption of the lemma, and the
  definition of $\alg$,
%  \begin{equation*}
    $\|\vec{A}(\vec{x} - \vec{x}')\|_p \leq \|\vec{A}\vec{x} -
    \tilde{\vec{y}}\|_p + \|\vec{A}\vec{x'} - \tilde{\vec{y}}\|_p <
  \disc_{p, \alpha}(\vec{A}),$
%  \end{equation*}
  and we've reached a contradiction. 
\end{proof}

Corollary~\ref{cor:mainlb}, instantiated with $p=2$, and
Lemmas~\ref{lm:hyperplanes-disc}--\ref{lm:boxes-highdim-disc} imply an
array of noise lower bounds for approximating geometric range counting
while satisfying $(\eps, \delta)$-differential privacy.

\begin{theorem}
  Any mechanism $\mech$ that, for any $P$ in $\mathbb{R}^d$ with $|P|
  = n$ and $d = O(1)$, with constant probability approximates the
  halfspace range counting problem within average squared error
  $o(n^{1 - 1/d})$ is not $(\eps, \delta)$-differentially private for
  any constant $\eps$ and $\delta$.
\end{theorem}

\begin{theorem}
  Any mechanism $\mech$ that, for any $P$ in $\mathbb{R}^d$ with $|P|
  = n$ and $d = O(1)$, with constant probability approximates the
  orthogonal range counting problem within average squared error
  $o((\log n)^{d - 1})$ is not $(\eps, \delta)$-differentially private
  for any constant $\eps$ and $\delta$.
\end{theorem}

\begin{theorem}
  Any mechanism $\mech$ that, for any $P$ in $\mathbb{R}^d$ with $|P|
  = n$ and $d = \Theta(\log n)$, with constant probability
  approximates the orthogonal range counting problem within average
  squared error $n^{o(1)}$ is not $(\eps, \delta)$-differentially
  private for any constant $\eps$ and $\delta$.
\end{theorem}

We also note that that Corollary~\ref{cor:mainlb},
instantiated with $p=\infty$ and Lemma~\ref{lm:gen-disc} imply a lower
bound on the worst case squared error for privately approximating
$m$ arbitrary range counting queries where $m$ is much larger than
$n$. 

\begin{theorem}\label{thm:gen-lb}
  Any mechanism $\mech$ that, for any range space $(P, \range)$ ($|P|
  = n$, $|\range| = m$), with constant probability approximates range
  counts for $\range$ with worst case squared error $o(n\log 2m/n)$ is
  not $(\eps, \delta)$-differentially private for any constant $\eps$
  and $\delta$.
\end{theorem}

The results of  Dinur and Nissim~\cite{Dinur2003} for $m = \O(n)$ and
$m = 2^n$ are special cases of Theorem~\ref{thm:gen-lb}. To the best
of our knowledge, this is the first lower bound that explicitly accounts for the
dependence of error on $m$ for arbitrary $m > n$. 

\section{Algorithm for Bounded Shatter Function Systems}
\label{sect:ub}

In this section we present an efficient (for constant $d$) $(\eps,
\delta)$-differentially private range counting algorithm for range
spaces with bounded shatter function. We prove the algorithm gives
optimal average squared error and almost optimal worst-case squared
error bounds. The algorithm is based on a novel use of a decomposition
that was first constructed by Matou\u{s}ek~\cite{matousek1995tight} to
prove optimal discrepancy upper bounds for bounded shatter function
range spaces. Even a careful application of known methods in
differential privacy together with the decomposition does not provide
optimal error bounds directly; we, however, prove that privacy can be
satisfied for a constant fraction of $P$ while achieving optimal error
bounds; then we recurse on the remainder of $P$. Aside from the
decomposition, this method of satisfying privacy for a fraction of the
database is inspired by partial coloring methods in discrepancy
theory.

We will make an essential use of the following lemma, due originally
to Haussler. The lemma bounds the size of an epsilon net in the hamming
metric. 
\begin{lemma}[\cite{haussler1995sphere}]
  \label{lm:packing}
  Let $(P, \range)$ be a range space with shatter function $\pi_\range(s) =
  O(s^d)$. Let $\Delta$ be an integer less than $|P|$. Let $\mathcal{S}
  \subseteq \range$ be a collection of ranges such that for any two
  ranges $R_1, R_2 \in \mathcal{S}$, the symmetric difference between
  $R_1$ and $R_2$ is at least $\Delta$. Then, $|\mathcal{S}| =
  O((|P|/\Delta)^d)$. 
\end{lemma}

We construct collections of ranges with large pairwise distance
$\Delta$ for gemetrically growing values of $\Delta$. Using the
collections as finer and finer epsilon nets, we can represent each
range in $\range$ as the union and set difference of smaller and
smaller ranges, while Lemma~\ref{lm:packing} allows us to control the
number of such ranges needed for each value of $\Delta$. We then
approximate range counts for the ranges that make up the
decomposition; the trade-off between range size and number of distinct
ranges allows us to balance the noise incurred by randomized response
and by using composition (Lemma~\ref{lm:composition}).

We first detail the construction. Our presentation
follows~\cite{matousek2010geometric}. Let $(P, \range)$ be a range
space with shatter function $\pi_\range(s) = O(s^d)$. Let $k = \lceil
\log_2 n \rceil$. For each $i \in \{0, \ldots, k\}$, let
$\mathcal{S}_i \subseteq \range$ be a maximal collection of ranges
such that the symmetric difference between any two ranges $R_1, R_2
\in \mathcal{S}_i$ is at least $n2^{-i}$. In particular,
$\mathcal{S}_k = \range$ and $\mathcal{S}_0 = \{\emptyset\}$. For each
$R \in \mathcal{S}_i$, fix a $R' \in \mathcal{S}_{i-1}$ such that the
symmetric difference between $R$ and $R'$ is at most $n2^{-i+1}$ (such
a range exists by maximality of $\mathcal{S}_{i-1}$). Then we set
$F(R) = R\setminus R'$ and $G(R) = R' \setminus R$, so that $R' = (R
\setminus F(R)) \cup G(R)$, $F(R) \subseteq R$, and $G(R) \cap (R
\setminus F(R)) = \emptyset$. Define a new collection of ranges
$\mathcal{T}_i = \{F(R), G(R): R \in \mathcal{S}_i\}$. We can start
from $R\in \range = \mathcal{S}_k$ and apply the construction
recursively, until we have
%  \begin{equation*}
    $\emptyset = ((\ldots((R \setminus F_k)\cup G_k)\ldots)\cup G_2)\setminus F_1,$
%  \end{equation*}
  where $F_i, G_i \in \mathcal{T}_i$. Bactracking to reconstruct $R$,
  we get
  \begin{equation}
    \label{eq:decomp}
    R = ((\ldots(F_1 \setminus G_2) \cup F_2\ldots) \setminus G_k)\cup
    F_k.
  \end{equation}
  All union operations are on disjoint sets and any set is subtracted
  from a set that entirely contains it. 

  Each range in $\mathcal{T}_i$ has size at most $n2^{-i+1}$ by
  construction; by Lemma~\ref{lm:packing}, $|\mathcal{S}_i| =
  O(2^{di})$, and, since each range in $\mathcal{S}_i$ corresponds to
  at most two ranges in $\mathcal{T}_i$, we also have $\mathcal{T}_i =
  O(2^{di})$ .  Let $\vec{T}^i$ be the incidence matrix of $\mathcal{T}_i$. The
  following lemma follows from the decomposition (\ref{eq:decomp}):
  \begin{lemma}
    \label{lm:decomp}
    Let $(P, \range)$ be a range space with $|P| = n$ and shatter
    function $\pi_\range(s) = O(s^d)$. Let $\vec{A}$ be the incidence matrix
    of $\range$. Then, there exist matrices $\vec{T}^i \in \{0,
    1\}^{s_i \times n}$ and $\vec{Q}^i \in \{0, \pm 1\}^{m \times s_i}$ such that
%  \begin{equation*}
    $\vec{A} = \sum_{i = 1}^k{\vec{Q}^i \vec{T}^i}$.
%  \end{equation*}
  Furthermore, we have the following properties for $\vec{T}^i$ and
  $\vec{Q}^i$:
  \vspace{-5pt}
  \begin{itemize}
    \setlength{\itemsep}{0.5pt}
    \setlength{\parskip}{0pt}
    \setlength{\parsep}{0pt}
  \item each row in $\vec{T}^i$ has at most $n2^{-i+1}$ nonzero entries;
  \item $s_i \leq C 2^{di}$ for some absolute constant $C$;
  \item each row in $\vec{Q}^i$ has at most 2 nonzero entries.
  \end{itemize}
  \end{lemma}

  For the degree of a point $p \in P$ in the range space
  $\mathcal{T}_i$, we use the notation $d_i(p) = |\{R\in
  \mathcal{T}_i: p \in R\}|$.

  Intuitively, we will use randomized response on those
  $\mathcal{T}_i$ consisting of only small ranges, and we will use the
  Laplace noise mechanism on those $\mathcal{T}_i$ consisting of few
  ranges. The ``breaking-even point'' for the analysis is $i_0 = (\log
  n)/d$. For $i \geq i_0$ randomized response gives the
  guarantee we need: the largest range in $\mathcal{T}_{i}$ for $i
  \geq i_0$ has size
  at most $n^{1 - 1/d}$. However, $\mathcal{T}_{i_0}$ can have as many
  as $n$ ranges, and it seems that we cannot use Laplace noise with
  variance $n^{1 - 1/d}$ and still preserve privacy for those $i$
  close to $i_0$. To circumvent this issue, we use the fact that we
  can bound both the largest range and the number of ranges in each
  $\mathcal{T}_i$ simultaneously.  The main observation is that we can
  add noise with optimal variance $O(n^{1 - 1/d})$ to the range counts for
  those $\mathcal{T}_i$ where randomized response doesn't work, and
  bound the average privacy loss $\frac{1}{n} \sum_p
  L_\mech(p)$. Then, we use averaging and Lemma~\ref{lm:composition}, and argue that we can preserve
  privacy for \emph{most} $p \in P$. The shatter function bound does not increase
  for restrictions of $P$ and $\range$ and we can recurse on the
  remaining points of $P$. \junk{This argument is insipired by partial
    coloring methods in discrepancy theory and is novel in the theory
    of differential privacy. }Our algorithm for computing range counts
  over ranges with bounded shatter function is given as
  Algorithm~\ref{alg:rangecount}. The algorithm description and the
  following discussion assume that $\range$ has shatter function
  $\pi_\range(s) = O(s^d)$ (for $d \geq 2$) and the decomposition of
  Lemma~\ref{lm:decomp} has already been computed. Note that the
  decomposition can be computed in time $O(mn\log n)$.

\begin{algorithm}[t]
  \caption{\textsc{RangeCount}($P, \vec{x}, \range, \eps, \delta$)}\label{alg:rangecount}
{\fontsize{9}{9}\selectfont
  \begin{algorithmic}
    \STATE Let $|P| = n$, $|\range| = m$;
    \STATE Set $i_0 := \frac{\log n}{d}$;
    \STATE Set $\eps_i := \frac{\eps(i - i_0 + 1)^{1.5}}{n^{1/2 - 1/2d}}$
    for $i \leq i_0$; 
    \STATE Set $\eps_i :=  \frac{\eps}{(i - i_0 + 1)^{1.5}}$ for $i > i_0$; 
    \IF {$n \leq 1$}
         \STATE Let $p \in P$ be the only point in $P$. Return $\tilde{x}_p
         := x_p + \Lap(1/\eps) $ for all $R \in \range$ s.t.~$p \in R$ and
         0  for all other $R \in \range$. 
    \ENDIF
    \STATE
    \STATE Set $X := \{p:\sum_{i \leq i_0}{d_i(p)\eps_i^2} \leq
    12C\eps^2\}$ and $\bar{X} := P \setminus X$;
    \STATE Recursively compute \textsc{RangeCount}($\bar{X}, \vec{x}|_{\bar{X}},
    \range|_{\bar{X}}, \eps, \delta$); let the results be
    $\tilde{z}^1_1, \ldots, \tilde{z}^1_m$.
    \FORALL{$i \leq i_0$}
        \STATE Compute $\tilde{\vec{y}}^i := (\vec{T}^i|_X)(\vec{x}|_X) +
        \Lap(1/\eps_i)^{s_i}$;
    \ENDFOR
    \FORALL{$i_0<i\leq k$}
        \STATE Compute $\tilde{\vec{x}}^i := \vec{x} + \Lap(1/\eps_i)^n$;
        \STATE Compute $\tilde{\vec{y}}^i := (\vec{T}^i|_X)(\tilde{\vec{x}}^i|_X)$;
    \ENDFOR
    \STATE
    \STATE Compute  $\vec{\tilde{z}}^2 := \sum_{i = 1}^k{\vec{Q}^i \vec{\tilde{y}}^i}$;
    \STATE Output $\tilde{\vec{z}} = \tilde{\vec{z}}^1 + \tilde{\vec{z}}^2$.
  \end{algorithmic}
}
\end{algorithm}

We analyze the privacy guarantees of
Algorithm~\ref{alg:rangecount}. We first prove some technical claims
about the algorithm.
\begin{lemma}
  \label{priv-lemma}
  The following hold for Algorithm~\ref{alg:rangecount}:

  \begin{enumerate}
  \item   $|X| \geq n/2$.
  \item   $\{\tilde{\vec{y}}^i\}_{i = 1}^{i_0}$ is a
    $(2\sqrt{6C}\eps\sqrt{\ln(1/\delta)},  \delta)$-differentially
    private function of $\vec{x}|_X$.
  \item   $\{\tilde{\vec{x}}^i\}_{i = i_0 + 1}^k$ is a $(2\eps,
  0)$-differentially private function of $\vec{x}$. Moreover, for each
  $S \subseteq P$, $\{\tilde{\vec{x}}^i|_S\}_{i = i_0 + 1}^k$ is a
  $(2\eps, 0)$-differentially private function of $\vec{x}|_S$.
  \end{enumerate}
\end{lemma}
\begin{proof}
  Claim 1.~follows by avaraging and the inequality
  \begin{equation}
    \label{eq:avgprivloss}
    \frac{1}{n}\sum_{p \in P}{\sum_{i < i_0}{d_i(p)\eps_i^2}} \leq 6C\eps^2
  \end{equation}
  Next we establish (\ref{eq:avgprivloss}).
  \begin{align*}
    \frac{1}{n}\sum_{p \in P}{\sum_{i \leq i_0}{d_i(p)\eps_i^2}} %&= \frac{1}{n}\sum_{i \leq i_0}\sum_{j =   1}^{s_i}{|\supp(\vec{T}^i_{j*})|\eps_i^2}\\
    &\leq     \frac{1}{n}\sum_{i \leq      i_0}{C2^{di}n2^{-i+1}\frac{\eps^2(i - i_0 + 1)^3}{n^{1 -   1/d}}}\\
    %&= \frac{C\eps^2}{n^{1 - 1/d}} \sum_{i \leq    i_0}{2^{di-i+1}(i - i_0 + 1)^3}\\ 
    &\leq 2C\eps^2 \sum_{j =  0}^{\infty}{\frac{(j+1)^3}{2^{dj + j}}}
    \leq 6C\eps^2. 
  \end{align*}
  The first inequality follows from Lemma~\ref{lm:decomp}. The second
  inequality holds for $d \geq 2$. This finishes the
  proof of claim 1.

  The following privacy analysis uses the fact that the range space
  $(P, \range)$ is public, and, therefore, the decomposition given by
  Lemma~\ref{lm:decomp}, and the set $X$ determined by the
  decomposition are public as well, i.e.~independent of $\vec{x}$.

  Notice that each component of $\tilde{\vec{y}}^i$ is an instance of
  the Laplace noise mechanism and, therefore, by
  Lemma~\ref{lm:laplace} it is $(\eps_i, 0)$-differentially
  private. Also, $\tilde{y}^i_j$ is independent of $x_p$ whenever
  $T^i_{jp} = 0$ or $p \not \in X$. Denoting by
  $\tilde{y}^i_j(\vec{x})$ the random variable $\tilde{y}^i_j$ when
  the input is $\vec{x}$, we have that
  \begin{equation*}
    D_\infty(\tilde{y}^i_j(\vec{x}) \| \tilde{y}^i_j(\vec{x} \pm
  \vec{e_p})) \leq
  \begin{cases}
    0, &T^i_{jp} = 0 \text{ or } p \not \in X\\
    \eps_i, &\text{otherwise}
  \end{cases}
  \end{equation*}
  If $\mech$ is the mechanism that outputs $\{\vec{\tilde{y}}^i\}_{i =
    1}^{i_0}$, then, by the above discussion, $L_\mech(p) =
  \sqrt{\sum_{i \leq i_0}{d_i(p)\eps_i^2}}$. By the definition of $X$,
  we have that $L_\mech(p) \leq \sqrt{12C}\eps$ for any $p \in X$ (and
  $L_\mech(p) = 0$ for $p \not \in X$).  Claim 2.~then follows by
  Lemma~\ref{lm:composition}.

  By Lemma~\ref{lm:laplace}, each
  $\tilde{\vec{x}}^i$ is $(\eps_i, 0)$-differentially private. By \junk{the
  basic composition theorem for differential
  privacy} Lemma~\ref{lm:simple_comp}, the composition
  $\{\tilde{\vec{x}}^i\}_{i = i_0 + 1}^k$ is $(\sum_{i = i_0 +
    1}^k{\eps_i}, 0)$-differentially private. Then claim 3.~follows
  from
  \begin{equation*}
    \sum_{i = i_0 + 1}^k{\eps_i} < \eps \sum_{j = 2}^\infty{j^{-1.5}} < 2\eps.
  \end{equation*}
  This completes the proof of the lemma.
\end{proof}

\begin{theorem}[\textbf{Privacy}]
  Algorithm~\ref{alg:rangecount} preserves $((2\sqrt{6C}+2)\eps\sqrt{\ln 1/\delta},
  \delta)$-differential privacy.
\end{theorem}
\begin{proof}
  We proceed by induction on $n$. 

  \textbf{Base case}. When $n \leq 1$, the output of
  Algorithm~\ref{alg:rangecount} is $(\eps, 0)$-differentially private,
  since it is a function of $\tilde{\vec{x}}$, which is itself $(\eps,
  0)$-differentially private by the properties of the Laplace noise
  mechanism (Lemma~\ref{lm:laplace}).
  
  \textbf{Inductive step}. Note that $\tilde{\vec{z}}^2$ is a function
  of $\tilde{\vec{x}}|_{X}$ and $\{\tilde{\vec{y}}^i\}_{i =
    1}^{i_o}$. Also note that both $\tilde{\vec{x}}|_{X}$ and
  $\{\tilde{\vec{y}}^i\}_{i = 1}^{i_o}$ depend only on $X$ and not on
  $\bar{X}$. By simple composition (Lemma~\ref{lm:simple_comp}), and
  Lemma~\ref{priv-lemma}, $\tilde{\vec{z}}^2$
  is a $((2\sqrt{6C}+2)\eps\sqrt{\ln 1/\delta}, \delta)$-differentially private
  function of $\vec{x}|_{X}$. By Lemma~\ref{priv-lemma}, $\bar{X} <
  n/2$, so by the inductive hypothesis $\tilde{\vec{z}}^1$ is an
  $((2\sqrt{6C}+2)\eps\sqrt{\ln 1/\delta}, \delta)$-differentially private function
  of $\vec{x}|_{\bar{X}}$. Since $X$ and $\bar{X}$ are disjoint, it
  follows that $\tilde{\vec{z}} = \tilde{\vec{z}}^1 +
  \tilde{\vec{z}}^2$ is a $(6(\sqrt{C}+2)\sqrt{\ln 1/\delta},
  \delta)$-differentially private function of $\vec{x}$.     
\end{proof}

Next we analyze the approximation guarantee of the algorithm. The
bounds in following lemma can derived by a straightforward
calculation. 
\begin{lemma}
  \label{util-lemma}
  Let $\vec{y}^i = (\vec{T}^i|_X)(\vec{x}|_X)$. For each $j \in [m]$
  and each $i \leq i_0$, $\E[\vec{Q}^i_{j*}\tilde{\vec{y}}^i] =
  \vec{Q}^i_{j*}\vec{y}^i$, and
  $\Var[\vec{Q}^i_{j*}\tilde{\vec{y}}^i] = O(n^{1 - 1/d}/(\eps^2(i -
  i_0 + 1)^3))$. 

  Similarly, for each $j \in [m]$ and each $i > i_0$,
  $\E[\vec{Q}^i_{j*}\tilde{\vec{y}}^i] =  \vec{Q}^i_{j*}\vec{y}^i$,
  and  $\Var[\vec{Q}^i_{j*}\tilde{\vec{y}}^i] = O(n^{1 - 1/d}(i - i_0
  + 1)^3/(2^{i-i_0}\eps^2))$.
\end{lemma}

We're now ready to prove an approximation guarantee.
\begin{theorem}[\textbf{Utility}]
  The expected average squared error of Algorithm~\ref{alg:rangecount}
  is $O(n^{1 - 1/d}/\eps^2)$. With probability at least $1-\beta$,
  the worst-case squared error of Algorithm~\ref{alg:rangecount} is at most
  $O(n^{1 -1/d}\log(n/\beta)/\eps^2)$.
\end{theorem}
\begin{proof}
  Let $\vec{z}^2 = \sum_{i = 1}^k{\vec{Q}^i\vec{y}^i}$. Note that all
  $\tilde{\vec{y}}^i$ have indepedentent noise. Then, by
  Lemma~\ref{util-lemma}, for each
  $j \in [m]$, $\E[\tilde{z}_j^2] = z_j^2$ and $\Var[\tilde{z}_j^2] =
  O(n^{1- 1/d}/\eps^2)$
  The expected total squared error of Algorithm~\ref{alg:rangecount} is,
  by linearity of expectation $\sum_j{\Var[\tilde{z}_j]}$. Since
  $\vec{\tilde{z}}^1$ is independent from $\tilde{\vec{z}}^2$, we have
%  \begin{equation*}
  $\sum_j{\Var[\tilde{z}_j]} = \sum_j{\Var[\tilde{z}^1_j]} +
  \sum_j{\Var[\tilde{z}^2_j]}$.  
%  \end{equation*}
  By claim 1.~in Lemma~\ref{priv-lemma}, the first term is the result of a
  recursive call on input of size at most $n/2$. We can express the
  expected squared error as a function $E(n)$ recursively as $E(n) =
  E(n/2) + O(n^{1-1/d}/\eps^2)$ which is easily seen to resolve to
  $E(n) = O(n^{1 - 1/d}/\eps^2)$. 

  The worst-case guarantee can be derived by standard use of tail
  bounds for sums of Laplace random variables.
\end{proof}

\section{Extensions}
\label{sect:extensions}

%\subsection{Other Constructions}
Algorithms for halfspace range counting can be derived from several
other methods, each of which provides weaker noise guarantees and/or less
generality.

The partition trees of Chan~\cite{chan2010optimal} imply a way to
factor the incidence matrix $\vec{A}$ of a range space induced by
$d$-dimensional halfspaces into matrices $\vec{Q}$ and $\vec{D}$ such
that $\vec{A} = \vec{QD}$, each column in $\vec{D}$ has at most
$O(\log \log n)$ nonzero elements, each row in $\vec{Q}$ has at most
$O(n^{1 - 1/d})$ nonzero elements, and $\vec{Q}$ and $\vec{D}$ both
have elements bounded in absolute value by $1$. Using
Lemma~\ref{lm:composition}, we can add Laplace noise with variance
$O(\frac{1}{\eps^2}\log \log n)$ to each element of $\vec{Dx}$,
preserving $(\eps\sqrt{\ln 1/\delta}, \delta)$ privacy. We can then
bound the variance of this mechanism to argue that, with constant
probability, the average squared error is $O(\frac{1}{\eps^2}n^{1 -
  1/d}\log \log n)$ and the worst case squared error is
$O(\frac{1}{\eps^2}n^{1 - 1/d}\log n \log \log n)$.

Welzl~\cite{Welzl:1992:STL:647823.736532}, and Chazelle and Welzl~\cite{Chazelle:1989:QRS:82362.82366} gave an algorithm that, given a set of points $P$
in $\mathbb{R}^d$, computes a spanning path such that any hyperplane
intersects the path in at most $O(n^{1 - 1/d})$ components. Then the
intersection of any halfspace with $P$ can be represented as the union
of $O(n^{1 - 1/d})$ disjoint intervals on the spanning path. An
algorithm for privately computing interval counting queries, e.g.~the
algorithm from~\cite{chan2010private}, can be used with the spanning
path as input, giving average squared error $O(\frac{1}{\eps^2}n^{1 -
  1/d}\log n)$ and worst case squared error $O(\frac{1}{\eps^2}n^{1 -
  1/d}\log^2 n)$. Interestingly, the spanning path approach generalizes
to range spaces whose \emph{dual} shatter function is bounded by a
polynomial with exponent $d$.

There is a well-known connection between combinatorial discrepancy and
epsilon approximations (c.f.~\cite{matousek2010geometric}, Chapter
1). Let $(P, \range)$ be a range space such that the maximum
discrepancy over all restrictions of $\range$ to a size $s$ subset of
$P$ is $f(s)$ (this is the same $f(s)$ as in Section~\ref{sect:defn}). Under some reasonable assumptions on the range space,
there exists a subset $S$ of $P$ of size $s$ such that range counts on
$S$ are close to range counts on $P$ to within an additive
$\frac{n}{s}f(s)$. Using this fact, and the
discrepancy upper bound for range spaces with shatter function
exponent $d$, we can apply the median mechanism of Roth and
Roughgarden~\cite{Roth2010} with the new analysis in~\cite{Gupta2011a}
to obtain a squared error upper bound that depends on $n$ as $O(n^{2d/(2d +
  1)})$. This upper bound is suboptimal; for example, for $d = 2$, it
yields an upper bound of $n^{4/5}$ as opposed to the optimal
$n^{1/2}$. Nevertheless, this method still gives squared error bounds that
grow slower than $n$ for range system with polynomial
shatter function. It also extends to the case where the universe
is much larger than $\|\vec{x}\|_1$. Giving optimal or near optimal
error upper bounds in this large universe regime is an interesting
open problem.
\junk{
\subsection{The Large Universe Case}

The special case of the range count problem when the universe $P$ is
much larger than the ``database size'' $\|\vec{x}\|_1$ (assumed to be
public information in this case) has attracted special attention in
the privacy literature.  Algorithm~\ref{alg:rangecount} gives a tight
noise upper for ranges with bounded shatter function bound when the
universe size is equal or almost equal to the database size. In this
section, we show how to privately reduce the universe size to be
almost linear in the database size. 

Only for this section, let us denote the universe size $|P| = N$ and
the database size $\|\vec{x}\|_1 = n$. Our main result in this section
applies beyond range spaces with bounded shatter function\junk{ and in fact
allows reducing the universe size to be almost linear in the database
size}. 

\begin{algorithm}
  \caption{\textsc{ExpSample}($P, \vec{x}, \eps, \beta$)}\label{alg:rejsample}
  \begin{algorithmic}
    \STATE Set $k = n\left(1 + \frac{2}{\eps}\ln \frac{N}{\beta}\right) + 1$;
    \STATE $P' = \emptyset$;
    \WHILE{$|P'| < k$}
       \STATE Sample $p \in P$ uniformly at random;
       \IF{$p \in P'$}
          \STATE Reject and continue;
       \ELSIF{$x_p \neq 0$}
          \STATE add $p$ to $P'$ with probability $e^\eps/(1+e^\eps)$;
       \ELSE
          \STATE add $p$ to $P'$ with probability $1/(1+e^\eps)$;
       \ENDIF
    \ENDWHILE
  \end{algorithmic}
\end{algorithm}

\begin{theorem}
  \label{thm:large2small}
  Given a set $P$ such that $|P| = N$ and an input $\vec{x}$ such that
  $\|\vec{x}\|_1 \leq n$ (where both $N$ and $n$ are public
  parameters), there exists an $(\eps, 0)$-differentially private
  mechanism $\mech$ that with probability $1-\beta$ output a set $P'
  \subseteq P$ such that
  \begin{itemize}
  \item $|P'| = O(\frac{n}{\eps}\log \frac{N}{\beta})$
  \item $x_p \neq 0 \Rightarrow p \in P'$
  \end{itemize}
  Moreover, $\mech$ runs in time $\poly(n, \log N)$ when given access
  to the nonzero elements of $\vec{x}$ in a compact representation and
  an oracle that provides uniform random samples from $P$.
\end{theorem}
\begin{proof}
  We will use the exponential mechanism to choose the subset $P'$ of
  $P$. Let us leave $k = |P'|$ as a parameter for now. Given a subset
  $P'$, we define the quality function of the exponential mechanism as
  \begin{equation}
    \label{eq:quality}
    q(P', \vec{x}) := |\{p \in P': x_p = 0\}|.
  \end{equation}
  At least $k - q(P', \vec{x})$ components of $x|_{P'}$ are nonzero,
  but, since $\vec{x}$ is an integer vector, there are at most $n$
  points $p$ such that $x_p \neq 0$. Therefore, if $n < k - q(P',
  \vec{x})$, then for all $p \in P$ such that $x_p \neq 0$, we have $p
  \in P'$. It remains to find a value for $k$ such that $n < k - q(P',
  \vec{x})$ with probability at least $1-\beta$.

  There are at least ${N \choose k-n}$ sets $P'$ of size $k$ that
  achieve the optimal quaility function value $q(P', \vec{x}) = 0$. 
  By Lemma~\ref{lm:exp-util}, with probability $1-\beta$, the
  exponential mechanism, when run with privacy parameter $\eps$,
  samples a set $P'$ such that 
  \begin{align}
    q(P', \vec{x}) &< \frac{2}{\eps}\left[\ln {N \choose
      k} - \ln {N \choose k-n} +  \ln\frac{1}{\beta}\right]\\ 
  &< \frac{2}{\eps}n\ln \frac{N}{\beta}.
  \end{align}
  Therefore, with probability $1 - \beta$,
  \begin{equation}
    k - q(P', \vec{x}) > k - \frac{2}{\eps}n\ln \frac{N}{\beta}.
  \end{equation}
  Taking $k > n(1 + \frac{2}{\eps}\ln \frac{N}{\beta})$ we have that
  with probability at least $1 - \beta$, $k - q(P', \vec{x}) > n$, as
  desired. 
  
  Next we argue that the exponential mechanism with quality function
  as defined in (\ref{eq:quality}) can be implemented in polynomial
  time. One implementation, based on rejection sampling is described
  in Algorithm~\ref{alg:rejsample}. Counting each sampling as a
  constant time operation,  the expected running
  time of the algorithm is $O(|P'|)$: each sampled point is rejected
  with probability $\Omega(1)$. Also, it's not hard to verify that
  Algorithm~\ref{alg:rejsample} samples from the desired distribution. 
\end{proof}

The fact that $\mech$ of Theorem~\ref{thm:large2small} runs in
polynomial time does not contradict the hardness results
of~\cite{Dwork:2009:CDP:1536414.1536467}, because of the oracle
assumption. It is easy to check that standard cryptographic
assumptions (i.e. existence of one-way functions) do not permit the
existence of a sampling oracle with running time sublinear in $N$ for
the hard-to-sanitize databases constructed in
in~\cite{Dwork:2009:CDP:1536414.1536467}. Nevertheless, we feel that
it is reasonable to assume that such an oracle exists for natural
geometric databases. For example, when $P$ is a $d$-dimensional grid,
we can sample uniformly in time $O(d\log N)$.
}

\section{Concluding Remarks}
While predicate count queries ($\vec{A}\vec{x}$)  have been studied in
differential privacy before, we make one of the first significant
progress in understanding the complexity of the problem in terms of
the combinatorial properties of $\vec{A}$, in particular for
halfspace, orthogonal and other range count queries. Our main result
is tight upper and lower bounds on approximation of $(\epsilon, \delta)$
differentially private halfspace count queries. Our approach is via a
variation of discrepancy. The main problems we leave open are to get tight
bounds for orthogonal counts with $(\epsilon, \delta)$-differential
privacy and to extend our bounds to the large universe regime.

\section*{Acknowledgements} 

We would like to thank Guy Rothblum, Kobbi Nissim, and Aaron Roth for
helpful discussions, and the anonymous reviewers for useful
suggestions and corrections.

This material is based upon work supported by the National Science
Foundation under Grant No.~0916782.

\bibliographystyle{plain}
\bibliography{privacy}

\begin{thebibliography}{10}

\bibitem{agarwal1999geometric}
P.K. Agarwal and J.~Erickson.
\newblock Geometric range searching and its relatives.
\newblock In {\em Advances in discrete and computational geometry: proceedings
  of the 1996 AMS-IMS-SIAM joint summer research conference, Discrete and
  Computational Geometry--Ten Years Later, July 14-18, 1996, Mount Holyoke
  College}, volume 223, page~1. Amer Mathematical Society, 1999.

\bibitem{Barak:2007:PAC:1265530.1265569}
Boaz Barak, Kamalika Chaudhuri, Cynthia Dwork, Satyen Kale, Frank McSherry, and
  Kunal Talwar.
\newblock Privacy, accuracy, and consistency too: a holistic solution to
  contingency table release.
\newblock In {\em Proceedings of the twenty-sixth ACM SIGMOD-SIGACT-SIGART
  symposium on Principles of database systems}, PODS '07, pages 273--282, New
  York, NY, USA, 2007. ACM.

\bibitem{Beck:fk}
Jozsef Beck.
\newblock Balanced two-colorings of finite sets in the square.
\newblock {\em Combinatorica}, 1(4):327--335, 1981.

\bibitem{Blum2008}
Avrim Blum, Katrina Ligett, and Aaron Roth.
\newblock A learning theory approach to non-interactive database privacy.
\newblock In {\em Proceedings of the 40th annual ACM symposium on Theory of
  computing}, STOC '08, pages 609--618, New York, NY, USA, 2008. ACM.

\bibitem{chan2010private}
T.H.H. Chan, E.~Shi, and D.~Song.
\newblock Private and continual release of statistics.
\newblock In {\em ICALP}, 2010.

\bibitem{chan2010optimal}
T.M. Chan.
\newblock Optimal partition trees.
\newblock In {\em Proceedings of the 2010 annual symposium on Computational
  geometry}, pages 1--10. ACM, 2010.

\bibitem{chazelle2000discrepancy}
B.~Chazelle.
\newblock {\em The discrepancy method}.
\newblock Cambridge Univ. Press, 2000.

\bibitem{chazelle2001trace}
B.~Chazelle and A.~Lvov.
\newblock A trace bound for the hereditary discrepancy.
\newblock {\em Discrete \& Computational Geometry}, 26(2):221--231, 2001.

\bibitem{chazelle1995elementary}
B.~Chazelle, J.~Matou{\u{s}}ek, and M.~Sharir.
\newblock An elementary approach to lower bounds in geometric discrepancy.
\newblock {\em Discrete \& Computational Geometry}, 13(1):363--381, 1995.

\bibitem{Chazelle:1989:QRS:82362.82366}
B.~Chazelle and E.~Welzl.
\newblock Quasi-optimal range searching in spaces of finite vc-dimension.
\newblock {\em Discrete Comput. Geom.}, 4:467--489, September 1989.

\bibitem{De2011}
Anindya De.
\newblock Lower bounds in differential privacy.
\newblock {\em CoRR}, abs/1107.2183, 2011.

\bibitem{Dinur2003}
Irit Dinur and Kobbi Nissim.
\newblock Revealing information while preserving privacy.
\newblock In {\em Proceedings of the twenty-second ACM SIGMOD-SIGACT-SIGART
  symposium on Principles of database systems}, PODS '03, pages 202--210, New
  York, NY, USA, 2003. ACM.

\bibitem{DMNS}
C.~Dwork, F.~Mcsherry, K.~Nissim, and A.~Smith.
\newblock Calibrating noise to sensitivity in private data analysis.
\newblock In {\em TCC}, 2006.

\bibitem{Dwork2010}
C.~Dwork, G.~N. Rothblum, and S.~Vadhan.
\newblock Boosting and differential privacy.
\newblock In {\em Proc. 51st Annual IEEE Symp. Foundations of Computer Science
  (FOCS)}, pages 51--60, 2010.

\bibitem{Dwork2007}
Cynthia Dwork, Frank McSherry, and Kunal Talwar.
\newblock The price of privacy and the limits of lp decoding.
\newblock In {\em STOC}, pages 85--94, 2007.

\bibitem{Dwork:2009:CDP:1536414.1536467}
Cynthia Dwork, Moni Naor, Omer Reingold, Guy~N. Rothblum, and Salil Vadhan.
\newblock On the complexity of differentially private data release: efficient
  algorithms and hardness results.
\newblock In {\em Proceedings of the 41st annual ACM symposium on Theory of
  computing}, STOC '09, pages 381--390, New York, NY, USA, 2009. ACM.

\bibitem{Dwork2008}
Cynthia Dwork and Sergey Yekhanin.
\newblock New efficient attacks on statistical disclosure control mechanisms.
\newblock In {\em CRYPTO}, pages 469--480, 2008.

\bibitem{Gupta2011}
Anupam Gupta, Moritz Hardt, Aaron Roth, and Jonathan Ullman.
\newblock Privately releasing conjunctions and the statistical query barrier.
\newblock In {\em Proceedings of the 43rd annual ACM symposium on Theory of
  computing}, STOC '11, pages 803--812, New York, NY, USA, 2011. ACM.

\bibitem{Gupta2011a}
Anupam Gupta, Aaron Roth, and Jonathan Ullman.
\newblock Iterative constructions and private data release.
\newblock {\em CoRR}, abs/1107.3731, 2011.

\bibitem{Hardt2010}
M.~Hardt and G.~N. Rothblum.
\newblock A multiplicative weights mechanism for privacy-preserving data
  analysis.
\newblock In {\em Proc. 51st Annual IEEE Symp. Foundations of Computer Science
  (FOCS)}, pages 61--70, 2010.

\bibitem{Hardt2010a}
Moritz Hardt, Katrina Ligett, and Frank McSherry.
\newblock A simple and practical algorithm for differentially private data
  release.
\newblock {\em CoRR}, abs/1012.4763, 2010.

\bibitem{Hardt:2010:GDP:1806689.1806786}
Moritz Hardt and Kunal Talwar.
\newblock On the geometry of differential privacy.
\newblock In {\em Proceedings of the 42nd ACM symposium on Theory of
  computing}, STOC '10, pages 705--714, New York, NY, USA, 2010. ACM.

\bibitem{haussler1995sphere}
D.~Haussler.
\newblock Sphere packing numbers for subsets of the boolean n-cube with bounded
  vapnik-chervonenkis dimension.
\newblock {\em Journal of Combinatorial Theory, Series A}, 69(2):217--232,
  1995.

\bibitem{Kasiviswanathan2010}
Shiva~Prasad Kasiviswanathan, Mark Rudelson, Adam Smith, and Jonathan Ullman.
\newblock The price of privately releasing contingency tables and the spectra
  of random matrices with correlated rows.
\newblock In {\em Proceedings of the 42nd ACM symposium on Theory of
  computing}, STOC '10, pages 775--784, New York, NY, USA, 2010. ACM.

\bibitem{matousek1995tight}
J.~Matou{\u{s}}ek.
\newblock Tight upper bounds for the discrepancy of half-spaces.
\newblock {\em Discrete and Computational Geometry}, 13(1):593--601, 1995.

\bibitem{matousek2010geometric}
J.~Matou{\u{s}}ek.
\newblock {\em Geometric discrepancy: An illustrated guide}, volume~18.
\newblock Springer Verlag, 2010.

\bibitem{Roth2010a}
Aaron Roth.
\newblock Differential privacy and the fat-shattering dimension of linear
  queries.
\newblock In {\em Proceedings of the 13th international conference on
  Approximation, and 14 the International conference on Randomization, and
  combinatorial optimization: algorithms and techniques}, APPROX/RANDOM'10,
  pages 683--695, Berlin, Heidelberg, 2010. Springer-Verlag.

\bibitem{Roth2010}
Aaron Roth and Tim Roughgarden.
\newblock Interactive privacy via the median mechanism.
\newblock In {\em Proceedings of the 42nd ACM symposium on Theory of
  computing}, STOC '10, pages 765--774, New York, NY, USA, 2010. ACM.

\bibitem{roth1954irregularities}
K.F. Roth.
\newblock On irregularities of distribution.
\newblock {\em Mathematika}, 1(02):73--79, 1954.

\bibitem{spencer1985six}
J.~Spencer.
\newblock Six standard deviations suffice.
\newblock {\em Trans. Amer. Math. Soc}, 289, 1985.

\bibitem{Welzl:1992:STL:647823.736532}
Emo Welzl.
\newblock On spanning trees with low crossing numbers.
\newblock In {\em Data Structures and Efficient Algorithms, Final Report on the
  DFG Special Joint Initiative}, pages 233--249, London, UK, 1992.
  Springer-Verlag.

\bibitem{Xiao2010}
Xiaokui Xiao, Guozhang Wang, and J.~Gehrke.
\newblock Differential privacy via wavelet transforms.
\newblock In {\em Proc. IEEE 26th Int Data Engineering (ICDE) Conf}, pages
  225--236, 2010.

\bibitem{YaoY85}
Andrew Chi-Chih Yao and F.~Frances Yao.
\newblock A general approach to d-dimensional geometric queries (extended
  abstract).
\newblock In {\em STOC}, pages 163--168, 1985.

\end{thebibliography}
\end{document}